\documentclass{article}
\usepackage{authblk}
\usepackage{graphicx,float}
\usepackage{pdfsync}
\usepackage[latin1]{inputenc}
\usepackage{mathrsfs,xspace}
\usepackage{bm,amsmath,amssymb}
\usepackage{textcomp}
\usepackage{algorithmic}
\usepackage[normalsize]{subfigure}
\usepackage{sidecap}

\newtheorem{lemma}{Lemma} 
\newtheorem{theorem}{Theorem}

\newtheorem{proof}{Proof}
\floatstyle{ruled}
\newfloat{Algorithm}{htpb}{alg}

\newcommand{\exit}[1]{\operatorname{exit}(#1)}

\newcommand{\T}{\operatorname{T}}
\newcommand{\HT}{\operatorname{HT}}

\newcommand{\lst}[2]{${#1}_0$,~${#1}_1$, $\dots\,$,~${#1}_{#2-1}$}

\newcounter{prgline}

\newtheorem{prop}{Proposition}

\newcommand{\rank}{\operatorname{rank}}

\newcommand{\select}{\operatorname{select}}

\def\..{\,\mathpunct{\ldotp\ldotp}} % Middle stuff for intervals. Usage: \..

\newcommand{\url}{\cite{myurl}}

\renewcommand{\epsilon}{\varepsilon}
\renewcommand{\phi}{\varphi}

\title{Fast Prefix Search in Little Space, with Applications\footnote{This article 
was presented at the 18th Annual European Symposium on Algorithms (ESA), Liverpool (United Kingdom), September 6-8, 2010. This version contains the appendices omitted from the version published in the conference proceedings.}}

\author[1]{Djamal Belazzougui}
\author[2]{Paolo Boldi}
\author[3]{Rasmus Pagh}
\author[2]{Sebastiano Vigna}

\affil[1]{Universit\'e Paris Diderot---Paris 7, France.}
\affil[2]{Dipartimento di Scienze dell'Informazione, Universit\`a degli
Studi di Milano, Italy.}
\affil[3]{IT University of Copenhagen, Denmark.}

\date{}

\begin{document}

\maketitle

\thispagestyle{empty}

\bibliographystyle{plain}

\begin{abstract}
It has been shown in the indexing literature that there is an essential difference 
between prefix/range searches on the one hand, and predecessor/rank searches on the other hand, in that the former provably allows faster query resolution.
Traditionally, prefix search is solved by data structures that are also dictionaries---they actually contain the strings in $S$.
For very large collections stored in slow-access memory, we propose much more
compact data structures that support \emph{weak} prefix searches---they return
the ranks of matching strings provided that \emph{some} string in $S$ starts with the given
prefix. In fact, we show that our most space-efficient data structure is
asymptotically space-optimal.

Previously, data structures such as String B-trees (and more complicated
cache-oblivious string data structures) have implicitly supported weak prefix
queries, but they all have query time that grows logarithmically with the size of
the string collection. In contrast, our data structures are simple, naturally
cache-efficient, and have query time that depends only on the length of the
prefix, all the way down to constant query time for strings that fit in one
machine word.

We give several applications of weak prefix searches, including exact prefix counting 
and approximate counting of tuples matching conjunctive prefix conditions. 
\end{abstract}

%\begin{abstract}
%A \emph{prefix search} returns the strings out of a given collection $S$ that
%start with a given prefix. Traditionally, prefix search is solved by data structures
%that are also dictionaries---they actually contain the strings in $S$.
%For very large collections stored in slow-access memory, we propose extremely
%compact data structures that solve \emph{weak} prefix searches---they return
%the correct result only if \emph{some} string in $S$ starts with the given
%prefix. Our first structure uses $\HT(S)+O(n\log\log \ell)$ bits of space, where
%$\ell$ is the average string length, and $\HT(-)$ is a data-aware measure
%related to the trie size of $S$. In the worst case,
%$\HT(S)=O(|S|\log \ell)$, so our structure is logarithmically smaller than a
%dictionary. When queried with a prefix $p$, our structure answers in time
% $O(|p|/w + \log|p|)$, where $w$ is the size of a machine word. 
%Checking that the weak search actually returns a correct result 
%requires at most one probe to the original data. We provide a matching space
%lower bound, showing that the structure is space optimal. Our second structure 
%is time optimal: it needs space $O(n\ell^{1/c}\log\ell)$, for any $c>0$, and
%answers in time $O(|p|/w)$.
%\end{abstract}

%\newpage

%\setcounter{page}{1}

\section{Introduction}

In this paper we are interested in the following problem (hereafter referred to
as \emph{prefix search}): given a collection of strings, find all the strings
that start with a given prefix. In particular, we will be interested in the
space/time tradeoffs needed to do prefix search in a static context (i.e., when
the collection does not change over time).

% Prefix searches are common in database applications, for example in the context
% of OLAP queries where they correspond to selecting values in a subtree of a
% dimension hierarchy. Another instance of such data are (unordered) XML trees,
% which can be identified with the set of strings encoding the nodes on all
% root-to-leaf paths. A prefix query corresponds to a simple XPath query of the
% form {\tt a/b/c/\dots}.
% 
There is a large literature on indexing of string collections. We refer to
Ferragina et al.~\cite{FGGSCSCCO,BFKCOSBT} for state-of-the-art results, with
emphasis on the cache-oblivious model. Roughly speaking, results can be divided
into two categories based on the power of queries allowed. As shown by
P\v{a}tra\c{s}cu and Thorup~\cite{patrascu07randpred} any data structure for bit
strings that supports predecessor (or rank) queries must either use super-linear
space, or use time $\Omega(\log |p|)$ for a query on a prefix $p$. On the other
hand, it is known that prefix queries, and more generally range queries, can be
answered in constant time using linear space~\cite{range-1D}.

Another distinction is between data structures where the query time grows with
the number of strings in the collection (typically comparison-based), versus
those where the query time depends only on the length of the query string
(typically some kind of trie)\footnote{Obviously, one can also combine the two in
a single data structure.}. In this paper we fill a gap in the literature by
considering data structures for \emph{weak prefix search}, a relaxation of prefix
search, with query time depending only on the length of the query string. In a
weak prefix search we have the guarantee that the input $p$ is a prefix of some
string in the set, and we are only requested to output the ranks (in
lexicographic order) of the strings that have $p$ as prefix. Weak prefix searches
have previously been implicitly supported by a number of string indexes, most
notably the String B-tree~\cite{Ferragina:1999:SBN} and its descendants. In the
paper we also present a number of new applications, outlined at the end of the
introduction.

Our first result is that weak prefix search can be performed by accessing a data
structure that uses just $O(n\log \ell)$ bits, where $\ell$ is the average string
length. This is much less than the space of $n\ell$ bits used for the strings
themselves. We also show that this is the minimum possible space usage for any
such data structure, regardless of query time. We investigate different
time/space tradeoffs: At one end of this spectrum we have constant-time queries
(for prefixes that fit in $O(1)$ words), and still asymptotically vanishing space
usage for the index. At the other end, space is optimal and the
query time grows logarithmically with the length of the prefix. Precise
statements can be found in the technical overview below.

\paragraph{Motivation for smaller indexes.}
Traditionally, algorithmic complexity is studied in the so-called RAM model.
However, in recent years a discrepancy has been observed between that model and
the reality of modern computing hardware. In particular, the RAM model assumes
that the cost of memory access is uniform; however, current
architectures, including distributed ones, have strongly non-uniform access cost, and this trend seems to go on, see e.g.~\cite{conf/isca/HardavellasFFA09} for recent work in this direction.
Modern computer memory is composed of hierarchies where each level in the
hierarchy is both faster and smaller than the subsequent level. As a consequence,
we expect that reducing the size of the data structure will yield faster query
resolution. Our aim in reducing the space occupied by the data structure is to
improve the chance that the data structure will fit in the faster levels of the hierarchy.
This could have a significant impact on performance, e.g.~in cases where the plain storage of the strings does not fit in main memory. For
databases containing very long keys this is likely to happen (e.g., static
repositories of URLs, that are of utmost importance in the design of search
engines, can contain strings as long as one kilobyte). In such cases, reduction
of space usage from $O(n\ell)$ to $O(n\log\ell)$ bits can be significant.

By
studying the weak version of prefix search, we are able to separate clearly the space used by the
original data, and the space that is necessary to store an index. G\'al
and Miltersen~\cite{GaMCPCSDS} classify structures as \emph{systematic} and
\emph{non-systematic} depending on whether the original data is stored verbatim
or not. Our indices provide a result \emph{without} using the original data, and
in this sense our structures for weak prefix search are non-systematic. Observe,
however, that since those structures gives no guarantee on the result for
strings that are not prefixes of some string of the set, standard information-theoretical
lower bounds (based on the possibility of reconstructing the original set of
strings from the data structure) do not apply.

\paragraph{Technical overview.}
For simplicity we consider strings over a binary alphabet, but our methods
generalise to larger alphabets (the interested reader can refer to appendix~\ref{sec:alpha_ext} for discussion on this point). Our main result is that weak prefix search needs just
$O(|p|/w+ \log|p|)$ time and $O(n\log \ell)$ space in addition to the original
data, where $\ell$ is the average length of the strings, $p$ is
the query string, and $w$ is the machine word size. For strings of fixed length $w$, this reduces to query time $O(\log w)$ and
space $O(n\log w)$, and we show that the latter is \emph{optimal} regardless of
query time. Throughout the paper we strive to state all space
results in terms of $\ell$, and time results in terms of the length of the actual
query string $p$, as in a realistic setting (e.g., term dictionaries of a search
engine) string lengths might vary wildly, and queries might be issued that are
significantly shorter than the average (let alone maximum) string length.
%
%Ultimately a search will need to have access to the strings in order to retrieve
%them, which implies a space usage of $O(n\ell)$ bits (this happens to classical
%data structures for prefix search, such as tries). However, we present a data
%structure solving a simpler version of the prefix search problem, which we call
%\emph{weak prefix search}. In the weak version, we have the guarantee that the
%input $p$ is a prefix of some string in the set, and we are only requested to
%output the ranks (in lexicographic order) of the strings that have $p$ as prefix.
%The search is performed by accessing a data structure that uses just $O(n\log
%\ell)$ bits---much less than the space occupied by the strings themselves.
Actually, the data structure size depends on the \emph{hollow trie size} of the
set $S$---a data-aware measure related to the trie size~\cite{GHSV07} that is
much more precise than the bound $O(n\log\ell)$.  

Building on ideas from~\cite{range-1D}, we then give an $O(1+|p|/w)$ solution
(i.e., constant time for prefixes of length $O(w)$) that uses space $O(n
\ell^{1/c}\log \ell)$. %While perhaps less interesting from a practical viewpoint, the
This structure shows that weak prefix search is possible in constant time using
sublinear space. This data structure uses $O(1+|p|/B)$ I/Os in the
cache-oblivious model.

\paragraph{Comparison to related results.}
If we study the same problem in the I/O model or in the cache-oblivious model,
the nearest competitor is the String B-tree~\cite{Ferragina:1999:SBN}, and its cache-oblivious version~\cite{BFKCOSBT}. In the {\em static\/} case, the String
B-tree can be tuned to use $O(n\log \ell)$ bits by carefully encoding the
string pointers, and it has very good search performance with $O(\log_B(n)+|p|/B)$
I/Os per query (supporting all query types discussed in this paper). However, a search for $p$ inside the String B-tree may involve $\Omega(|p|)$ RAM
operations, so it may be too expensive for intensive
computations\footnote{Actually, the string B-tree can be tuned to work in
$O(|P|/w+\log n)$ time in the RAM model, but this would imply a
$O(|P|/B+\log n)$ I/O cost instead of $O(|P|/B+\log_B n)$.}. Our first
method, which also achieves the smallest possible space usage of $O(n\log \ell)$ bits, uses $O(|p|/w + \log |p|)$ RAM operations and $O(|p|/B+\log |p|)$ I/Os instead. The number of RAM operations is a strict improvement over String B-trees, while the I/O bound is better for large enough sets. Our second method
uses slightly more space ($O(n\ell^{1/c}\log\ell)$ bits) but features
$O(|p|/w)$ RAM operations and $O(|p|/B)$ I/Os.

In~\cite{FGGSCSCCO}, the authors discuss very succinct static data structures
for the same purposes (on a generic alphabet), decreasing the space to a lower
bound that is, in the binary case, the trie size. The search time is logarithmic in the number of strings. As in the previous case, we
improve on RAM operations and on I/Os for large enough sets.

The first cache-oblivious dictionary supporting prefix search was 
devised by Brodal \textit{et al.}~\cite{cacheobsd} achieving $O(|p|)$ RAM operations and $O(|p|/B)+\log_B(n)$ I/Os. 
We note that the result in~\cite{cacheobsd} is optimal in a comparison-based model,
where we have a lower bound of $\log_B(n)$ I/Os per query. By contrast, our
result, like those in~\cite{BFKCOSBT,FGGSCSCCO}, assumes an integer
alphabet where we do not have such a lower bound.

Implicit in the paper of Alstrup \textit{et al.}~\cite{range-1D} on range
queries is a linear-space structure for constant-time weak prefix search on
fixed-length bit strings. Our constant-time data structure, instead, uses
sublinear space and allows for variable-length strings.

\paragraph{Applications.}
Data structures that allow weak prefix search can be used to solve the non-weak
version of the problem, provided that the original data is stored (typically, in
some slow-access memory): a single probe is sufficient to determine if the result
set is empty; if not, access to the string set is needed just to retrieve the
strings that match the query. 
We also show how to solve range queries with two
additional probes to the original data (wrt.~the output size), improving the
results in~\cite{range-1D}. We also present other
applications of our data structures to other important problems, viz., prefix
counting.
\iffalse
, prefix sampling and prefix range minimal queries.
\fi 
We finally show that our results extend to the cache-oblivious model, where we provide an alternative to the 
results in ~\cite{cacheobsd,BFKCOSBT,FGGSCSCCO} that removes the dependence on the data set size 
for prefix searches and range queries.

\paragraph{Our contributions.}
The main contribution of this paper is
the identification of the weak
prefix search problem, and the proposal of an optimal solution based on
techniques developed in~\cite{BBPMMPH}. 
Optimality (in space or time) of the solution is also a central
result of this research. 
The second interesting contribution is the description
of \emph{range locators} for variable-length strings; they are an essential
building block in our weak prefix search algorithms, and can be used whenever 
it is necessary to recover in little space the range of leaves under a node of
a trie.

\section{Notation and tools}
\label{sec:notation}

In the following sections, we will use the toy set of strings shown in
Figure~\ref{fig:names} to display examples of our constructions.
In this section, we introduce some terminology and notation adopted throughout
the rest of the paper. We use von Neumann's definition
and notation for natural numbers: $n=\{\,0,1,\ldots,n-1\,\}$, so $2=\{\,0,1\,\}$
and $2^*$ is the set of all binary strings.

\noindent\textbf{Weak prefix search.} Given a prefix-free set of strings
$S\subseteq 2^*$, the \emph{weak prefix search} problem requires, given a
prefix $p$ of some string in $S$, to return the range of strings of $S$ having
$p$ as prefix; this set is returned as the interval of integers that
are the ranks (in lexicographic order) of the strings in $S$
having $p$ as prefix.

\noindent\textbf{Model and assumptions.} Our model of computation is a unit-cost
word RAM with word size $w$. We assume that $|S|=O(2^{cw})$ for some constant
$c$, so that constant-time static data structures depending on $|S|$ can be
used.
\\
We also consider bounds in the cache-oblivious model. In this model, we consider 
that the machine has a two levels memory hierarchy, where the fast level has an 
unknown size of $M$ bits (which is actually not used in this paper) and a slower
level of unspecified size where our data structures resides. We assume that
the slow level plays a role of cache for the fast level with an optimal replacement strategy 
where the transfers between the two levels is done in blocks of an unknown size
of $B$ bits, with $B\leq M$. The cost of an algorithm is the total number of
block transfers between the two levels. 
\\
\noindent\textbf{Compacted tries.} Consider the compacted trie built for a prefix-free set
of strings $S\subseteq 2^*$. For a given node $\alpha$ of the
trie, we define (see Figure~\ref{fig:names}):

\begin{figure}
\begin{center}
\begin{tabular}{cc}
\includegraphics[scale=.8]{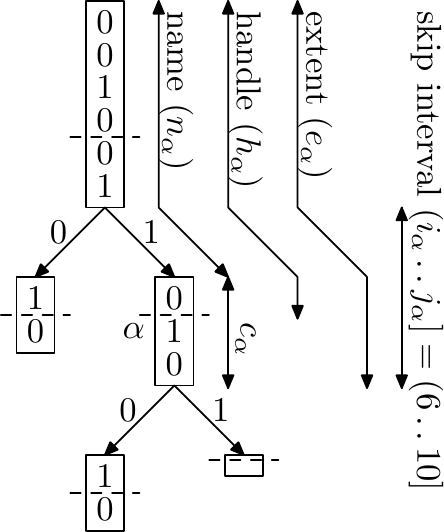} & \begin{minipage}{8cm}
$s_0$ = 001001010\\
$s_1$ = 0010011010010\\
$s_2$ = 00100110101\\
\end{minipage}
\end{tabular}
\caption{\label{fig:names}The compacted trie for the set
$S=\{\,s_0,s_1,s_2\,\}$, and the related names.}
\end{center}
\end{figure}

\begin{figure}
\begin{center}
\small $T$
\begin{tabular}{c}
\fbox{\begin{tabular}{lcl}
0 & $\to$ & $\infty$\\
00 & $\to$ & $\infty$\\
0010 & $\to$ & $001001\ (6)$\\
0010010 & $\to$ & $\infty$\\
00100101 & $\to$ & $001001010\ (9)$\\
0010011 & $\to$ & $\infty$\\
00100110 & $\to$ & $0010011010\ (10)$\\
00100110100 & $\to$ & $\infty$\\
001001101001 & $\to$ & $0010011010010\ (13)$\\
00100110101 & $\to$ & $00100110101\ (11)$\\
\end{tabular}}\qquad
\fbox{\begin{tabular}{l|c}
$P$ & $\bm b$\\
\hline
$\underline{001001}\mathbf{0}$ & 1\\
$\underline{001001\mathbf{1}}$ & 0\\
$\underline{001001101}0\mathbf{0}$ & 1\\
$\underline{0010011010\mathbf{1}}$ & 1\\
$\underline{0010011011}0$& 0\\
$\underline{00101}00 $ & 0
\end{tabular}}
\end{tabular}
\end{center}
\caption{\label{fig:ztrie}The data making up a z-fast prefix trie based
on the trie above, and the associated range locator. $T$ maps handles to extents;
the corresponding \emph{hollow} z-fast prefix trie just returns the
\emph{lengths} (shown in parentheses) of the extents. In the range locator
table, we boldface the zeroes and ones appended to extents, and we underline the
actual keys (as trailing zeroes are removed). The last two keys are $00100110101^+$ and $0010011^+$, respectively.}
%\caption{\label{fig:example}A toy example.}
\end{figure}

\begin{itemize}
	\item $e_\alpha$, the \emph{extent of node $\alpha$}, is the longest
	common prefix of the strings represented by the leaves 
	that are descendants of $\alpha$ (this was called the ``string represented by
	$\alpha$'' in~\cite{BBPMMPH});
	\item $c_\alpha$, the \emph{compacted path of node $\alpha$}, is the string stored at $\alpha$;
	\item $n_\alpha$, the \emph{name of node $\alpha$}, is the string $e_\alpha$
	deprived of its suffix $c_\alpha$ (this was called the ``path leading to
	$\alpha$'' in~\cite{BBPMMPH});
	\item given a string $x$, we let $\exit x$ be the exit node of $x$, that is,
	the only node $\alpha$ such that $n_\alpha$ is a prefix of $x$ and either $e_\alpha=x$ or $e_\alpha$
	is not a prefix of $x$;
	\item the \emph{skip interval} $[i_\alpha\..j_\alpha)$ associated to
	$\alpha$ is $[0\..|c_\alpha|)$ for the root, and $[|n_\alpha|-1\..|e_\alpha|)$
	for all other nodes.
\end{itemize}
We note the following property, proved in Appendix~\ref{sec:extentsproof}:
\begin{theorem}
\label{teo:extents}
The average length of the extents of internal nodes is at most the
average string length minus one.
\end{theorem}

\noindent\textbf{Data-aware measures.} Consider the compacted trie on a set
$S\subseteq 2^*$. We define the \emph{trie measure} of $S$~\cite{GHSV07} as
\[
\mathrm \T(S) = \sum_\alpha (j_\alpha - i_\alpha) = \sum_\alpha
(|c_\alpha| + 1) - 1 =2n - 2 + \sum_\alpha|c_\alpha| = O(n\ell),
\]
where the summation ranges over all nodes of the trie. For the purpose of this
paper, we will also use the \emph{hollow trie measure}
\[
\mathrm{HT}(S) = \sum_{\text{$\alpha$ internal}}(\operatorname{bitlength}(|c_\alpha|) + 1) - 1.
\]
Since $\operatorname{bitlength}(x)=\lceil \log(x + 1)\rceil$, we have
$\mathrm{HT}(S) = n - 2 +\sum_{\text{$\alpha$ internal}} \left\lceil
\log(|c_\alpha| +1)\right\rceil = O(n\log\ell)$.\footnote{A compacted trie is made
\emph{hollow} by replacing the
compacted path at each node by its length \emph{and then discarding all its leaves}. A recursive
definition of hollow trie appears in~\cite{BBPTPMMPH}.}

\noindent\textbf{Storing functions.} The problem of storing statically an $r$-bit
function $f: A \to 2^r$ from a given set of keys $A$ has recently received
renewed attention~\cite{DiPSDSRAM,ChCBFSL,PorOBFRBMS}. For the purposes of this
paper, we simply recall that these methods allow us to store an $r$-bit
function on $n$ keys using $rn+cn +o(n)$ bits for some constant $c\geq 0$, with
$O(|x|/w)$ access time for a query string $x$. Practical implementations are
described in~\cite{BBPTPMMPH}. In some cases, we will store a \emph{compressed}
function using a minimal perfect function ($O(n)$ bits) followed
by a compressed data representation (e.g., an Elias--Fano
compressed list~\cite{BBPTPMMPH}). In that case, storing natural numbers \lst
xn requires space $\sum_i \lfloor \log (x_i + 1)\rfloor + n \log( \sum_i
\lfloor \log ( x_i + 1)\rfloor/n)+ O(n)$.

\noindent\textbf{Relative dictionaries.} A \emph{relative dictionary} stores a
set $E$ relatively to some set $S\supseteq E$. That is, the relative dictionary
answers questions about membership to $E$, but its answers are required to be
correct only if the query string is in $S$. It is possible to store such a
dictionary in $|E|\log(|S|/|E|)$ bits of space with $O(|x|/w)$ access time~\cite{BBPMMPH}.

\noindent\textbf{Rank and select.} We will use two basic blocks of
several succinct data structures---rank and select. Given a bit array (or bit string) $\bm b\in 2^n$, whose
positions are numbered starting from 0, $\rank_{\bm b}(p)$ is the number of ones
up to position $p$, exclusive ($0\leq p\leq n$), whereas $\select_{\bm b}(r)$ is
the position of the $r$-th one in $\bm b$, with bits numbered starting from
$0$ ($0\leq r <\rank_{\bm b}(n)$). It is well known that these operations can be
performed in constant time on a string of $n$ bits using additional
$o(n)$ bits, see~\cite{Jac89,CM96,BM2,Raman07}.

\section{From prefixes to exit nodes}

We break the weak prefix search problem into two subproblems. Our first goal is
to go from a given a prefix of some string in $S$ to its exit node.

\subsection{Hollow z-fast prefix tries}

We start by describing an improvement of the \emph{z-fast trie}, a
data structure first defined in~\cite{BBPMMPH}.
The main idea behind a z-fast trie is that, instead of representing
explicitly a binary tree structure containing compacted paths of the trie, we 
will store a function that maps a certain prefix of each extent to the extent itself. This mapping (which can be stored in linear space) will be sufficient to navigate the trie and obtain, given a string $x$, the 
\emph{name of the exit node of $x$} and the exit behaviour (left,
right, or possibly equality for leaves). The interesting point about the z-fast
trie is that it provides such a name in time $O(|x|/w+\log |x|)$, and that it
leads easily to a probabilistically relaxed version, or even to blind/hollow variants.

To make the paper self-contained, we recall the main definitions
from~\cite{BBPMMPH}. The \emph{2-fattest} number in a nonempty interval of 
positive integers is the
number in the interval whose binary representation has the largest number of
trailing zeros. Consider the compacted trie on $S$, one of its nodes
$\alpha$, its skip interval $[i_\alpha\..j_\alpha)$, and the $2$-fattest
number $f$ in $(i_\alpha\.. j_\alpha]$ (note the change); if the
interval is empty, which can happen only at the root, we set $f=0$. 
The \emph{handle} $h_\alpha$ of $\alpha$ is $e_\alpha[0\..f)$, where
$e_\alpha[0\..f)$ denotes the first $f$ bits of $e_\alpha$. A
\emph{(deterministic) z-fast trie} is a dictionary $T$ mapping each handle $h_\alpha$ to the corresponding extent $e_\alpha$. In
Figure~\ref{fig:ztrie}, the part of the mapping $T$ with non-$\infty$ output is
the z-fast trie built on the trie of Figure~\ref{fig:names}.

We now introduce a more powerful structure, the \emph{(deterministic) z-fast
prefix trie}. Consider again a node $\alpha$ of the compacted trie on $S$ with 
notation as above. The \emph{pseudohandles} of $\alpha$ are the strings
$e_\alpha[0\..f')$, where $f'$ ranges among the 2-fattest numbers of the
intervals $(i_\alpha\..t]$, with $i_\alpha< t < f$. Essentially,
pseudohandles play the same r\^ole as handles for every \emph{prefix} of the
handle that extends the node name. We note immediately that there are
at most $\log(f-i_\alpha)\leq \log|c_\alpha|$ pseudohandles
associated with $\alpha$, so the overall number of handles and pseudohandles is
bounded by $\HT(S)+\sum_{x\in S}\log|x|=O(n\log \ell)$. It is now easy to define
a z-fast prefix trie: the dictionary providing the map from handles to extents is enlarged to pseudohandles, which are
mapped to the special value $\infty$.

We are actually interested in a \emph{hollow} version of a
z-fast prefix trie---more precisely, a version implemented by a function $T$
that maps handles of internal nodes to the length of their extents, and handles of leaves and pseudohandles to $\infty$.    
The function (see again Figure~\ref{fig:ztrie}) can be stored in a very small
amount of space; nonetheless, we will still be able to compute the name of
the exit node of any string that is a prefix of some string in $S$ using
Algorithm~2, whose correctness is proved in
Appendix~\ref{sec:correctnessproof}.

\begin{figure}
\begin{minipage}{.45\linewidth}
\begin{algorithmic}
\STATE{\textbf{Algorithm 1}}
\STATE{\textbf{Input: } a prefix $p$ of some string in~$S$.}
\STATE $i\leftarrow \lfloor\log |p|\rfloor$
\STATE $\ell,r\leftarrow 0,|p|$
\WHILE{$r-\ell>1$}
	\IF{$\exists b$ such that $2^ib	\in (\ell\..r)$}
	    \STATE // $2^ib$ is 2-fattest number in $(\ell\..r)$
		\STATE $g\leftarrow T\bigl(p[0\..2^ib)\bigr)$
		\IF{$g\geq|p|$} \STATE $r\leftarrow 2^ib$  
		\ELSE \STATE $\ell\leftarrow g$ \ENDIF
	\ENDIF
	\STATE $i\leftarrow i-1$
\ENDWHILE
\IF{$\ell = 0$} \STATE{{\bf return} $\epsilon$} 
\ELSE \STATE{{\bf return} $p[0\..\ell+1)$}
\ENDIF
\end{algorithmic}
\caption{\label{alg:findpar}Given a nonempty string $p$ that is the prefix of at
least one string in the set $S$, this algorithm returns the name of $\exit p$.}
\end{minipage}
\hspace{10mm}
%\end{figure}
%\begin{figure}
	\begin{minipage}{.45\linewidth}
\begin{algorithmic}
\STATE{\textbf{Algorithm 2}}
\STATE{\textbf{Input: } the name $x$ of a node}
\IF{$x=\epsilon$}
	\STATE{$i\leftarrow 0$, $j\leftarrow n$}
\ELSE
	\STATE $i\leftarrow\rank_{\bm b}{h(x^\leftarrow)}$
	\IF{$x=111\cdots11$}
		\STATE{$j\leftarrow n$}
	\ELSE
		\STATE $j\leftarrow\rank_{\bm b}{h((x^+)^\leftarrow)}$
	\ENDIF
\ENDIF
\RETURN $[i\..j)$
\end{algorithmic}
\caption{\label{alg:rquery}Given the name $x$ of a node in a trie containing
$n$ strings, compute the interval $[i\..j)$ containing precisely all the
(ranks of the) strings prefixed by $x$ (i.e., the strings in the subtree whose
name is $x$).}
\vspace{22mm}
\end{minipage}
\end{figure}

% Note that when queried with an input that is not a prefix of some string in
% $S$, the hollow z-fast prefix trie will return an impredictable result.

\subsection{Space and time}

The space needed for a hollow z-fast prefix trie depends on the
component chosen for its implementation. The most trivial bound uses a function mapping handles
and pseudohandles to one bit that makes it possible to recognise handles of
internal nodes ($O(n\log \ell)$ bits), and a function mapping handles to
extent lengths ($O(n\log L)$ bits, where $L$ is the maximum string
length).

These results, however, can be significantly improved. First
of all, we can store handles of internal nodes in a relative dictionary. The
dictionary will store $n-1$ strings out of $O(n\log\ell)$ strings, using
$O(n\log((n\log\ell)/n))=O(n\log\log\ell)$ bits. Then, the mapping from handles
to extent lengths $h_\alpha\mapsto|e_\alpha|$ can actually be recast into a mapping 
$h_\alpha\mapsto|e_\alpha|-|h_\alpha|$. But since
$|e_\alpha|-|h_\alpha|\leq|c_\alpha|$, by storing this data using a compressed
function we will
use space
\begin{multline*}
\sum_\alpha\lfloor\log(|e_\alpha|-|h_\alpha|+1)\rfloor+ O(n\log\log \ell)+
O(n)\\\leq \sum_\alpha\lfloor\log(|c_\alpha|+1)\rfloor + O(n\log\log \ell)\leq
\HT(S)+ O(n\log\log\ell),
\end{multline*}
where $\alpha$ ranges over internal nodes.

Algorithm~1 cannot iterate more than $\log|p|$ times; at each
step, we query constant-time data structures using a prefix of $p$: using
incremental hashing~\cite[Section~5]{DGMP}, we can preprocess $p$ in time
$O(|p|/w)$ (and in $|p|/B$ I/Os) so that hashing prefixes of $p$ requires
constant time afterwards. We conclude that Algorithm~1 requires time $O(|p|/w + \log|p|)$.

\subsection{Faster, faster, faster\ldots}
\label{sec:consttime}

We now describe a data structure mapping prefixes to exit nodes inspired by the
techniques used in~\cite{range-1D} that needs $O(n\ell^{1/2}\log \ell)$ bits of
space and answers in time $O(|p|/w)$, thus providing a different space/time
tradeoff. The basic idea is as follows: let $s =
\bigl\lceil\ell^{1/2}\bigr\rceil$ and, for each node $\alpha$ of the compacted
trie on the set $S$, consider the set of prefixes of $e_\alpha$ with
length $t\in(i_\alpha\..j_\alpha]$ such that either $t$ is a multiple of $s$ or
is smaller than the first such multiple. More precisely, we consider
prefixes whose length is either of the form $ks$, where $ks\in (i_\alpha\..j_\alpha]$, or in
$(i_\alpha\..\min\{\,\bar ks,j_\alpha\,\}]$, where $\bar k$ is the minimum $k$
such that $ks>i_\alpha$.
%More intuitively, we consider all the prefixes of
%$e_\alpha$ with lengths in $(i_\alpha\..j_\alpha]$ until we hit a length that is
%a multiple of $s$, and then we consider just the subsequent multiples of $s$.

We store a function $F$ mapping each prefix $p$ defined above to the length
of the name of the corresponding node $\alpha$ (actually, we can map $p$ to
$|p|-|n_\alpha|$). Additionally, we store a mapping $G$ from each node
name to the length of its extent (again, we can just map $n_\alpha \mapsto
|c_\alpha|$).

To retrieve the exit node of a string $p$ that is a
prefix of some string in $S$, we consider the string $q=p[0\..|p|-|p|\bmod s)$ (i.e.,
the longest prefix of $p$ whose length is a multiple of $s$). Then, we check whether 
$G(p[0\..F(q)))\geq|p|$ (i.e., whether $p$ is a prefix of the extent of the
exit node of $q$). If this is the case, then clearly $p$ has the same exit node
as $q$ (i.e., $p[0\..F(q))$). Otherwise, the map $F$ provides
directly the length of the name of the exit node of $p$, which is thus
$p[0\..F(p))$. All operations are completed in time
$O(|p|/w)$.

The proof that this structure uses space $O(n\ell^{1/2}\log \ell)$ is deferred
to Appendix~\ref{sec:spaceproof}.

\section{Range location}\label{sec:ranking}

Our next problem is determining the range (of lexicographical ranks) of the
leaves that appear under a certain node of a trie. Actually, this problem is
pretty common in static data structures, and usually it is solved by associating
with each node a pair of integers of $\log n\leq w$ bits. However, this means
that the structure has, in the worst case, a linear ($O(nw)$) dependency on the
data.

To work around this issue, we propose to use a \emph{range locator}---an
abstraction of a component used in~\cite{BBPMMPH}. Here we redefine range
locators from scratch, and improve their space usage so that it is dependent on
the average string length, rather than on the maximum string length.
%\footnote{Interestingly, 
%the structure is used therein to
%obtain monotone minimal perfect hashing in $O(n\log\log w)$ bits, but the
%structure itself requires a monotone minimal perfect hash function.} 

A range locator takes as input \emph{the name of a node}, and returns the
range of ranks of the leaves that appear under that node. For instance, in
our toy example the answer to $0010011$ would be $[1\..3)$.
To build a range locator, we need to introduce \emph{monotone minimal
perfect hashing}.

Given a set of $n$ strings $T$, a \emph{monotone minimal perfect
hash function}~\cite{BBPMMPH} is a bijection $T\to n$ that preserves
lexicographical ordering. This means that each string of $T$ is mapped to its rank in $T$ (but
strings not in $T$ give random results). We use the following results
from~\cite{BBPTPMMPH}:\footnote{Actually, results in~\cite{BBPTPMMPH} are
stated for prefix-free sets, but it is trivial to make a set of strings
prefix-free at the cost of doubling the average length.}
\begin{theorem}
\label{teo:mmph}
Let $T$ be a set of $n$ strings of average length $\ell$ and maximum length
$L$, and $x\in 2^*$ be a string. Then, there are monotone minimal perfect
hashing functions on $T$ that:
\begin{enumerate}
  \item use space $O(n\log\ell)$ and answer in time $O(|x|/w)$;
\item  use space $O(n\log\log L)$ and answer in time $O(|x|/w + \log |x|)$.
\end{enumerate}
\end{theorem}
We show how a reduction can relieve us from the dependency on $L$; this is
essential to our goals, as we want to depend just on the average length:
\begin{theorem}
\label{teo:mmmphavlen}
There is a monotone
minimal perfect hashing function on $T$ using space $O(n\log\log\ell)$ that
answers in time $O(|x|/w + \log |x|)$ on a query string $x\in 2^*$.
\end{theorem}
\begin{proof}
We divide $T$ into the set of strings $T^-$ shorter then
$\ell\log n$, and the remaining ``long'' strings $T^+$. Setting up a $n$-bit
vector $\bm b$ that records the elements of $T^-$ with select-one and select-zero
structures ($n+o(n)$ bits), we can reduce the problem to hashing monotonically $T^-$ and $T^+$. We note, however,
that using
Theorem~\ref{teo:mmph} $T^-$ can be hashed in space
$O(|T^-|\log\log(\ell\log n))=O(|T^-|\log\log\ell)$, as $2\ell\geq\log n$, and
$T^+$ can be hashed explicitly using a $(\log n)$-bit function; since $|T^+|\leq n/\log n$ necessarily, the function requires $O(n)$ bits.
Overall, we obtain the required bounds.
\end{proof}

We now describe in detail our range locator, using the notation of
Section~\ref{sec:notation}. Given a string $x$, let $x^\leftarrow$ be $x$ with
all its trailing zeroes removed. We build a set of strings $P$ as follows: for
each extent $e$ of an internal node, we add to $P$ the strings
$e^\leftarrow$, $e1$, and, if $e\neq
111\cdots11$, we also add to $P$ the string $(e1^+)^\leftarrow$, where $e1^+$
denotes the successor of length $|e1|$ of $e1$ in lexicographical order (numerically, it is $e1 + 1$).
We build a monotone minimal perfect hashing function $h$ on $P$, noting the
following easily proven fact:
\begin{prop}
The average length of the strings in $P$ is at most $3\ell$.
\end{prop}
% \begin{proof}
% By Theorem~\ref{teo:extents}, strings of the form $e1$ (which are all
% distinct) have average length at most $\ell$. When adding strings of the form
% $e^\leftarrow$ and $(e1^+)^\leftarrow$ to the set, each string is smaller in size than the
% corresponding $e1$ string. Thus, even if some strings collapse, the average
% length cannot exceed $3\ell$.
% \end{proof}

The second component of the range locator is a bit vector $\bm b$ of
length $|P|$, in which bits corresponding to the names of leaves are set to one. The
vector is endowed with a ranking structure $\rank_{\bm b}$ (see
Figure~\ref{fig:ztrie}).

It is now immediate that given a node name $x$, by hashing $x^\leftarrow$ and
ranking the bit position thus obtained in $\bm b$, we obtain the left extreme of
the range of leaves under $x$. Moreover, performing the same operations on
$(x^+)^\leftarrow$, we obtain the right extreme. All these strings are in $P$ by
construction, except for the case of a node name of the form $111\cdots11$; however, in that case the right extreme is
just the number of leaves (see Algorithm~\ref{alg:rquery} for the details).

%%% Algorithm 2 moved up

A range locator uses at most $3n+o(n)$ bits for $\bm b$ and its selection
structures. Thus, space usage is dominated by the monotone hashing component.
Using the structures described above, we obtain:
\begin{theorem}
There are structures implementing range location in time $O(|x|/w)$ using
$O(n\log \ell)$ bits of space, and in $O(|x|/w +\log |x|)$ time using
$O(n\log\log \ell)$ bits of space.
\end{theorem}

We remark that other combinations of monotone minimal perfect hashing and
succinct data structures can lead to similar results. For instance, we could
store the trie structure using a preorder standard balanced parentheses
representation, use hashing to retrieve the lexicographical rank $r$ of a node name, select
the $r$-th open parenthesis, find in constant time the matching closed
parenthesis and obtain in this way the number of leaves under the node. Among
several such asymptotically equivalent solutions, we believe ours is the most
practical.
% As an example, by replacing counters with a constant-time range locator
% in~\cite{range-1D}, we can immediately improve their result:
% \begin{theorem}
% There is a constant-time static data structure for range queries on strings
% of fixed length $w$ that uses $O(n w^{1/2}\log w)$ bits in addition to the
% original data.
% \end{theorem}
% In Section~\ref{sec:app} we will further improve this result by making it
% applicable to sets of variable-length strings.

\section{Putting It All Together}
\label{sec:prefix}

In this section we gather the main results about prefix search:
\begin{theorem}
\label{teo:logprefix}
There are structures implementing weak prefix search in space
$\HT(S)+O(n\log\log\ell)$ with query time $O(|p|/w+ \log|p|)$, and
in space
$O(n\ell^{1/2}\log\ell)$ with query time $O(|p|/w)$.
\end{theorem}
\begin{proof}
The first structure uses a hollow z-fast prefix trie followed by the range locator of
Theorem~\ref{teo:mmmphavlen}: the first component 
provides the name $n_\alpha$ of exit node of $|p|$; given $n_\alpha$, the range
locator returns the correct range.
For the second structure, we use the
structure defined in Section~\ref{sec:consttime} 
followed by the first range locator of Theorem~\ref{teo:mmph}.
\end{proof}

Actually, the second structure described in Theorem~\ref{teo:logprefix} can be
made to occupy space $O(n\ell^{1/c}\log\ell)$ for any constant $c>0$, as shown
in Appendix~\ref{sec:const2proof}:
\begin{theorem}
\label{teo:constprefix}
For any constant $c>0$, there is a structure implementing weak prefix search in
space $O(n\ell^{1/c}\log\ell)$ with query time
$O(|p|/w)$.
\end{theorem}
We note that all our time bounds can be translated into I/O bounds in the \emph{cache-oblivious model} 
if we replace the $O(|p|/w)$ terms by $O(|p|/B)$. The $O(|p|/w)$ term 
represents appears in two places: 
\begin{itemize}
\item The phase of precalculation of a hash-vector of $\lceil|p|/w\rceil$ hash words on the prefix $p$ which is later used to compute all the hash functions on prefixes of $p$. 
\item In the range location phase, where we need to compute $x^\leftarrow$ and $(x^+)^\leftarrow$, where $x$ is a prefix of $p$ and subsequently  compute the hash vectors on $x^\leftarrow$ and $(x^+)^\leftarrow$ . 
\end{itemize}
Observe that the above operations can be carried on using arithmetic operations only without any additional I/O (we can use $2$-wise independent hashing involving only multiplications and additions for computing the hash vectors and only basic basic arithmetic operations for computing $x^\leftarrow$ and $(x^+)^\leftarrow$) except for the writing the result of the computation which occupies $O(|p|/w)$ words of space and thus take $O(|p|/B)$ I/Os. Thus both of the two phases need only $O(|p|/B)$ I/Os corresponding to the time needed to read the pattern and to write the result. 
%As a consequence, in the cache oblivious model we need only a scanning of the prefix $p$ which takes $O(|p|/B)$ I/Os. Of particular interest is our constant time weak prefix search which takes optimal $O(|p|/B)$ I/Os in the cache-oblivious model.
 
\section{A space lower bound}

In this section we show that the space usage achieved by the weak prefix search
data structure described in Theorem~\ref{teo:logprefix} is optimal up to a
constant factor. In fact, we show a matching lower bound for the easier problem
of prefix counting (i.e., counting how many strings start with a given prefix),
and consider the more general case where the answer is only required to be
correct up to an additive constant less than $k$. We note that any data structure
supporting prefix counting can be used to achieve approximate prefix counting, by
building the data structure for the set that contains every $k$-th element in
sorted order. The proof is in Appendix~\ref{sec:lowerboundproof}.

\begin{theorem}\label{thm:lowerbound}
Consider a data structure (possibly randomised) indexing 
a set $S$ of $n$ strings with average length $\ell > \log(n) + 1$, supporting $k$-approximate prefix
count queries: Given a prefix of some key in $S$, the structure returns the
number of elements in $S$ that have this prefix with an additive error of less than $k$, where $k < n/2$. 
The data structure may return any number when given a string that is not a prefix of a key in $S$.  
Then the expected space usage on a worst-case set $S$ is $\Omega((n/k)\log(\ell - \log n))$ bits. 
In particular, if no error is allowed and $\ell >(1+\varepsilon)\log n$, for constant $\varepsilon > 0$, the expected space usage 
is $\Omega(n \log \ell)$ bits.
\end{theorem}
Note that the trivial information-theoretical lower bound does not apply, as
it is impossible to reconstruct $S$ from the data structure. 

It is interesting to note the connections with the lower and upper bounds
presented in~\cite{FGGSCSCCO}. This paper shows a lower bound on the number of bits
necessary to represent a set of strings $S$ that, in the binary case,
reduces to $\T(S)+\log\ell$, and provide a matching data structure.
Theorem~\ref{teo:logprefix} provides a \emph{hollow} data structure
that is sized following the naturally associated measure: $\HT(S)+O(n\log\log\ell)$. 
Thus, Theorem~\ref{teo:logprefix} and~\ref{thm:lowerbound} can be
seen as the hollow version of the results presented in~\cite{FGGSCSCCO}.
Improving Theorem~\ref{thm:lowerbound}
to $\HT(S)+o(\HT(S))$ is an interesting open problem.

\section{Applications}
\label{sec:app}

In this section we highlight some applications of weak prefix search. In several
cases, we have to access the original data, so we are actually using weak prefix
search as a component of a \emph{systematic} (in the sense of~\cite{GaMCPCSDS})
data structure. However, our space bounds consider only the indexing data
structure. Note that the pointers to a set of string of overall
$n\ell$ bits need in principle $O(n\log\ell)$ bits of spaces to be represented;
this space can be larger than some of the data structures themselves.
Most applications can be turned into cache-oblivious data structures, but this
discussion is postponed to the Appendix for the sake of space.

In general, we think that the space used to store and access the original data
should not be counted in the space used by weak/blind/hollow structures, as the
same data can be shared by many different structures. There is a standard
technique, however, that can be used to circumvent this problem: by using $2n\ell$
bits to store the set $S$, we can round the space used by each string to the
nearest power of two. As a results, pointers need just $O(n\log\log\ell)$ bits to
be represented.

\subsection{Prefix search and counting in minimal probes}
\label{sec:prefixsearch}

The structures for weak prefix search described in Section~\ref{sec:prefix} can
be adapted to solve the prefix search problem within the same bounds, provided
that the actual data are available, although typically in some slow-access memory. Given a prefix $p$ we
get an interval $[i\..j)$. If there
exists some string in the data set prefixed by $p$, then it should be at one of
the positions in interval $[i\..j)$, and all strings in that interval are
actually prefixed by $p$. So we have reduced the search to two alternatives:
either all (and only) strings at positions in $[i\..j)$ are prefixed by $p$, or
the table contains no string prefixed by $p$. This implies the two following
results:
\begin{itemize}
\item We can report all the strings prefixed by a prefix $p$ in optimal number
of probes. If the number of prefixed strings is $t$, then we will probe exactly $t$
positions in the table. If no string is prefixed by $p$, then we will probe a
single position in the table. 
\item We can count the number of strings prefixed
by a given prefix in just one probe: it
suffices to probe the table at any position in the interval $[i\..j)$: if the
returned string is prefixed by $p$, we can conclude that the number of strings
prefixed by $p$ is $j-i$; otherwise, we conclude
that no string is prefixed by $p$.
\end{itemize}

\subsection{Range emptiness and search with two additional probes}

The structures for weak prefix search described in Section~\ref{sec:prefix} can
also be used for range emptiness and search within the same bounds, again if the
actual data is available. In the first case, given two strings $a$ and $b$ we ask whether any
string in the interval $[a\..b]$ belongs to $S$; in the second case we must report all such strings. 

Let $p$ the longest common prefix of $a$ and $b$
(which can be computed in time $O(|p|/w)$). Then we have two sub-cases
\begin{itemize}
\item The case $p=a$ ($a$ is actually a prefix of $b$). We are looking 
for all strings prefixed by $a$ which are lexicographically smaller 
than $b$. We perform a prefix query for $a$, getting $[i\..j)$. Then we can
report all elements in $S\cap[a\..b]$ by doing a scan
strings at positions in $[i\..j)$ until we encounter a string which is not in 
interval $[a\..b]$. Clearly the number of probed positions is $|S\cap[a\..b]|+1$. 
\item The case $p\neq a$. We perform a prefix query for $p0$,
getting $[i_0\..j_0)$ and another query for $p1$, getting $[i_1\..j_1)$. Now it is
immediate that if $S\cap [a\..b]$ is not empty, then necessarily it is
made by a suffix of $[i_0\..j_0)$ and by a prefix of $[i_1\..j_1)$.
We can now report $S\cap[a\..b]$ using at most 
$|S\cap[a\..b]|+2$ probes; we start from the end of the first interval and scan
backwards until we find an element not in $[a\..b]$; then, we start from the
beginning of the second interval and scan forwards until we find an element not in
$[a\..b]$
\end{itemize}

We report all elements thus found: clearly, we make at most two
additional probes. In particular, we can report whether $S\cap[a\..b]$ is empty in at most two
probes. These results improve the space bounds of the index described
in~\cite{range-1D}, provide a new index using just $\HT(S)+O(n\log\log\ell)$
bits, and give bounds in terms of the average length.

\bibliography{prefix}

\newpage

\appendix

\section{Conclusions}
We have presented two data structures for prefix search that provide different
space/time tradeoffs. In one case (Theorem~\ref{teo:logprefix}), we prove a
lower bound showing that the structure is space optimal. In the other case
(Theorem~\ref{teo:constprefix}) the structure is time optimal. It is also
interesting to note that the space usage of the time-optimal data structure can
be made arbitrarily close to the lower bound. Our structures are based on \emph{range locators},
a general building block for static data structures, and on structures that are able to map prefixes to names of the
associated exit node. In particular, we discuss a variant of the z-fast trie, the
\emph{z-fast prefix trie}, that is suitable for prefix searches. Our variant
carries on the good properties of the z-fast trie (truly linear space and
logarithmic access time) and significantly widens its usefulness by making it
able to retrieve the name of the exit node of prefixes. We have shown
several applications in which sublinear indices access very quickly data in
slow-access memory, improving some results in the literature.

\section{Proof of Theorem~\ref{teo:extents}}\label{sec:extentsproof}
Let $E$ be the sum of the lengths of the extents of internal nodes, and $M$
the sum of the lengths of the strings in the trie.
We show equivalently that  $E\leq M(n-1)/n - n + 1$. This is obviously
true if $n=1$. Otherwise, let $r$ be the length of the compacted path at the root,
and let $n_0$, $n_1$ be the number of leaves in the left and right subtrie;
correspondingly, let $E_i$ the sum of lengths of the extents of each subtrie,
and $M_i$ the sum of the lengths of the strings in each subtrie, stripped of
their first $r+1$ bits. Assuming by induction $E_i\leq M_i(n_i-1)/n_i - n_i +
1$, we have to prove
\begin{multline*}
E_0+(n_0-1)(r+1)+E_1+(n_1-1)(r+1) + r \leq\\
\frac{n_0+n_1-1}{n_0+n_1}(M_0+M_1+(n_0+n_1)(r+1)) - n_0-n_1 +1,
\end{multline*}
which can be easily checked to be always true under the assumption above.

\section{Proof of the space bound claimed in
Section~\ref{sec:consttime}}\label{sec:spaceproof}

First of all, it can easily be proved that
the domain of $F$ is $O(n\ell^{1/2})$ in size. Each
$\alpha$ contributes at most $s$ prefixes whose lengths are in interval $(i_\alpha\..\min(\bar ks,j_\alpha)]$. It also contributes at most $(j_\alpha-i_\alpha)/s+1$ prefixes whose lengths are of the form $ks$, where $ks\in
(i_\alpha\..j_\alpha]$. Overall the total number of prefixes is no more than:
$$\sum_{\alpha}(s+(j_\alpha-i_\alpha)/s+1)=(s+1)(2n-1)+\sum_{\alpha}((j_\alpha-i_\alpha)/s)$$

The sum of lengths of skip intervals of all nodes of the trie $T(S)$ is no larger than sum of lengths 
of strings $n\ell$:  $$\T(S)=\sum_{\alpha}(j_\alpha-i_\alpha)\leq n\ell$$
 From that we have: 
$$\sum_{\alpha}((j_\alpha-i_\alpha)/s)=\frac{1}{s}\sum_{\alpha}(j_\alpha-i_\alpha)\leq \frac{1}{s}n\ell\leq ns$$
Summing up, the  total number of prefixes is less than $(s+1)(2n-1)+ns=O(ns)=O(n\ell^{1/2})$.
Since the output size of the function $F$ is bounded by
$\max_\alpha\log|c_\alpha|\leq\log L$, where $L$ is the \emph{maximum} string
length, we would obtain the space bound $O(n\ell^{1/2}\log L)$. To prove the
strict bound $O(n\ell^{1/2}\log \ell)$, we need to further refine the structure
so that the ``cutting step'' $s$ is larger in the deeper regions of the trie.

Let $S^-$ be the subset of strings of $S$ of length less than $\ell(\log n)^2$,
and $S^+$ the remaining strings. We will change the step $s$ after depth $\ell
(\log n)^2$. Let $s=\bigl\lceil\ell^{1/2}\bigr\rceil$ and let $s^+=\bigl\lceil(\ell (\log
n)^2)^{1/2}\bigr\rceil=\bigl\lceil{\ell}^{1/2}\log n\bigr\rceil$. We will say
that a node is \emph{deep} if its extent is long at least $\ell (\log
n)^2$. We will split $F$ into a function $F^-$ with output size $\log (\ell (\log n)^2)=\log \ell + 2\log \log n=O(\log \ell)$
that maps prefixes shorter than $\ell(\log n)^2$ (\emph{short prefixes}),
and a function $F^+$ with output size $\log (\ell n)=\log n + \log \ell$ that
maps the remaining \emph{long prefixes}. For every node $\alpha$ with skip
interval $[i_\alpha\..j_\alpha)$, we consider three cases:
\begin{enumerate}

\item If $j_\alpha< \ell (\log n)^2$ (a non-deep node), we will store the
prefixes of $e_\alpha$ that have lengths either of the form $ks$, where $ks\in(i_\alpha\..j_\alpha]$, or in
$(i_\alpha\..\min\{\,\bar ks,j_\alpha\,\}]$, where $\bar k$ is the minimum $k$
such that $ks>i_\alpha$. Those prefixes are short, so they will be mapped using
$F^-$.

\item If $i_\alpha<\ell (\log n)^2\leq j_\alpha$ (a deep node with non-deep
parent), we store the following prefixes of $e_\alpha$:
\begin{enumerate}
\item Prefixes of lengths $ks$, where $ks\in (i_\alpha\..\ell (\log n)^2)$, or of
lengths in $(i_\alpha\..\min\{\,\bar ks,\ell (\log n)^2\,\}]$, where $\bar k$ is
the minimum $k$ such that $ks>i_\alpha$. Those prefixes are short, so they
will be mapped using $F^-$.
\item Prefixes of lengths $\ell (\log n)^2+ks^+$, where  $\ell (\log
n)^2+ks^+\in [\ell(\log n)^2\..j_\alpha]$. Those prefixes are long, so they will
be mapped using $F^+$.
\end{enumerate}
\item If $i_\alpha\geq\ell (\log n)^2$ (a deep node with a deep parent), we will
store all prefixes that have lengths either of the form $\ell (\log
n)^2+ks^+$, where $\ell (\log n)^2+ks^+\in (i_\alpha\..j_\alpha]$, or in
$(i_\alpha\..\min\{\,\ell (\log n)^2+\bar ks^+,j_\alpha\,\}]$, where $\bar k$ is
the minimum $k$ such that $\ell (\log n)^2+ks^+>i_\alpha$. Those prefixes are long, so they will
be mapped using $F^+$.
\end{enumerate}

The function $F$ is now defined by combining $F^-$ and $F^+$ in the obvious way.
To retrieve the exit node of a string $p$ that is a
prefix of some string in $S$, we have two cases: if $|p|<\ell(\log n)^2$, we consider the string
$q=p[0\..|p|-|p|\bmod s)$, otherwise we consider the string
$q=p[0\..|p|-(|p|-\lceil\ell (\log n)^2\rceil)\bmod s^+)$. Then, we check
whether $G(p[0\..F(q)))\geq|p|$ (i.e., whether $p$ is a prefix of the extent of the exit
node of $q$). If this is the case, then we conclude clearly $p$ has the same exit
node of $q$ (i.e., $p[0\..F(q))$). Otherwise, the map $F$ gives the name of the
exit node of $p$ : $p[0\..F(p))$.

The space bound holds immediately for $F^-$, as we already showed that
prefixes (long and short) are overall $O(n\ell^{1/2})$, and $F^-$ has output
size $\log(\ell(\log n)^2))=O(\log \ell)$.

To bound the size of $F^+$, we first bound the number of deep nodes. Clearly
either a deep node is a leaf or it has two deep children. If a deep node
is a leaf then its extent is long at least
$\ell (\log n)^2$, so it represent a string from $S^+$. Hence,
the collection of deep nodes constitute a forest of complete binary trees where the number of leaves
is the number of strings in $S^+$. As the number of strings in $S^+$ is at most
$n/(\log n)^2$, we can conclude that the total number of nodes in the forests
(i.e., the number of deep nodes) is at most $2n/(\log n)^2-1$. For each deep
node we have two kinds of long prefixes:
\begin{enumerate}
\item Prefixes that have lengths of the form $\ell (\log n)^2+ks^+$. Those
prefixes can only be prefixes of long strings, and for each long string $x\in
S^+$, we can have at most $|x|/s^+$ such prefixes. As the total length of all
strings in $S^+$ is at most $n\ell$, we conclude that the total number of
such prefixes is at most $n\ell/s^+=O(n\ell/\bigl\lceil\ell^{1/2}\log
n\bigr\rceil)=O(n\ell^{1/2}/(\log n))$.
\item Prefixes that have lengths in $(i_\alpha\..\min\{\ell (\log
n)^2+\bar ks^+,j_\alpha\}]$ or in $[\ell(\log n)^2\..\min\{\ell (\log n)^2+\bar ks^+,j_\alpha\}]$ for a
node $\alpha$, where $\bar k$ is the minimum $k$ such that $\ell (\log
n)^2+ks^+> i_\alpha$ . We can have at most $s^+$ prefix per
node: since we have at most $2n/(\log n)^2-1$ nodes, the number of prefixes of
that form is $O(s^+n/(\log n)^2)=O(n\ell^{1/2}/(\log n))$.
\end{enumerate}
As we have a total of $O(n\ell^{1/2}/(\log n))$ long prefixes, and the output
size of $F^+$ is $O(\log \ell + \log n)$, we can conclude that total space used for $F^+$ is
bounded above by $O((\log \ell + \log n)n\ell^{1/2}/\log n)=O(n\ell^{1/2}\log \ell)$.

Finally, we remark that implementing $G$ by means of a compressed function we
need just $\HT(S)+O(n\log \log \ell) + O(n)=O(n\log\ell)$ bits of space.

\section{Correctness proof of
Algorithm~1}\label{sec:correctnessproof}

The correctness of Algorithm~1 is expressed by the following
\begin{lemma}\label{lem:correctness}
Let $X=\{\,p_0=\epsilon,p_1,\dots,p_t\,\}$, where $p_1$,~$p_2$,
$\dots\,$,~$p_t$ are the extents of the nodes of the trie that are
prefixes of $p$, ordered by increasing length. 
Let $(\ell\..r)$ be the interval maintained by the
algorithm. Before and after each iteration the following invariants 
are satisfied: 
\begin{enumerate}
    \item\label{enu:lem1} there exists at most a single $b$ such that $2^ib\in
    (\ell\..r)$;
    \item\label{enu:lem2} $\ell=|p_j|$ for some $j$, and $\ell\leq|p_t|$;
    \item\label{enu:lem3} during the algorithm, $T$ is only queried with strings
    that are either the handle of an ancestor of $\exit p$, or a pseudohandle
    of $\exit p$; so, $T$ is well-defined on all such strings;
    \item\label{enu:lem4} $|p_t|\leq r$;
\end{enumerate}
\end{lemma}

We will use the following property of 2-fattest numbers, proved in~\cite{BBPMMPH}:

\begin{lemma}
\label{lemma:dyadic}
  Given an interval $[x\..y]$ of strictly positive integers:
  \begin{enumerate}
    \item\label{enu:dyadicone} Let $i$ be the largest number such that
  there exists an integer $b$ satisfying $2^ib\in [x,y]$. Then $b$ is unique,
  and the number $2^ib$ is the \emph{2-fattest number in $[x\..y]$}.    
    \item\label{enu:dyadictwo} If $y-x< 2^i$, there exists at most a 
  single value $b$ such that $2^ib\in [x\..y]$.    
    \item\label{enu:dyadicthree} If $i$ is such that $[x\..y]$ does not contain
    any value of the form $2^ib$, then $y-x+1 \leq 2^i-1$ and the interval may
    contain at most one single value of the form $2^{i-1}b$.
  \end{enumerate} 
\end{lemma}

Now, the correctness proof of Lemma~\ref{lem:correctness} follows.
\noindent(\ref{enu:lem1}) Initially, when $i=\lfloor\log |p|\rfloor$ we have
  $(\ell\..r)=(0\..|p|)$, and this interval contains at most a single
  value of the form $2^ib$, that is $2^i$.
 Now after some iteration suppose that we have at most a 
 single $b$ such that $2^ib\in (\ell\..r)$. 
  We have two cases:
  \begin{itemize}
  \item There is no $b$ such that $2^i b\in (\ell\..r)$. Then, the
  interval remains unchanged and, by 
  Lemma~\ref{lemma:dyadic} (\ref{enu:dyadicthree}), it will contain at
  most a single value of the form $2^{i-1}b$.
  \item There is a single $b$ such that $2^i b\in (\ell\..r)$.
  The interval may be updated in two ways: either we set the interval to 
  $(g\..r)$ for some $g\geq 2^i b$ or we set the interval 
  to $(\ell\..2^ib)$. In both cases, the new interval will no longer contain
  $2^ib$. By invariant \ref{enu:dyadicthree}.~of Lemma~\ref{lemma:dyadic}, the new interval
  will contain at most a single value  of the form
  $2^{i-1}b$.
  \end{itemize}

\noindent(\ref{enu:lem2})
The fact that $\ell=|p_j|$ for some $j$ is true at the beginning, and when
$\ell$ is reassigned it remains true: indeed, if
$T\bigl(p[0\..2^ib)\bigr)=g<|p|$ this means that $p[0\..2^ib)$ is the handle of
a node $\alpha$ found on the path to $\exit p$, and $g=|e_\alpha|$; but since
$p$ is a prefix of some string, $p[0\..g)=e_\alpha$ and the latter is $p_j$ for
some $j$. This fact also implies $\ell\leq|p_t|$, since the $p_i$'s have decreasing
lengths. 

\noindent(\ref{enu:lem3})
By (\ref{enu:lem2}), $\ell$ is always the length of the extent of some $p_j$,
whereas $r=|p|$ at the beginning, and then it can only decrease; so $(\ell\..r)$
is a union of some skip intervals of the ancestors of $\exit p$ and of an
initial part of the skip interval of the node $\exit p$ itself. Hence, its
2-fattest number is either the handle of some of the ancestors (possibly, of
$\exit p$ itself) or a pseudohandle of $\exit p$ (this can only happen if $r$
is not larger than the 2-fattest number of the skip interval of $\exit p$).

\noindent(\ref{enu:lem4}) 
The property is true at the beginning. Then, $r$ is reduced only in two cases:
either $2^ib$ is the 2-fattest number of the skip interval of $\exit p$
(in this case, $g$ is assigned $|e_{\exit p}|\geq |p|$); or we are
querying with a pseudohandle of $\exit p$ or with the handle of a leaf
(in this case, $g$ is assigned the value $\infty$). In both cases, $r$ is
reduced to $2^ib$, which is still not smaller than the extent of the parent of
$\exit p$.

%\iffalse
\section{Proof of Theorem~\ref{teo:constprefix}}\label{sec:const2proof}
Rather that describing the proof from scratch, we describe the changes
necessary to the proof given in Appendix~\ref{sec:spaceproof}.

The main idea is that of setting $t=\lceil \ell^{1/c}\rceil$, $t^+=\lceil
\ell^{1/c}\log n\rceil$, and let $s=t^{c-1}$ and $s^+={t^+}^{c-1}$.
In this way, since $n\ell=O(nt^c)$ clearly the prefixes of the form $ks$ and
$\ell (\log n)^c+ks^+$ are $O(n\ell^{1/c})$ and $O(n\ell^{1/c}/\log n)$. The problem is that now
the prefixes at the start of each node (i.e., those of length $i_\alpha$, $i_\alpha+1$,~\ldots) are too many.

To obviate to this problem, we record significantly less prefixes. More
precisely we record sequences of prefixes of increasing lengths: 
\begin{itemize}
\item For non deep nodes we first store 
prefixes of lengths $i_\alpha$, $i_\alpha+1$,~\ldots\ until
we hit a multiple of $t$, say $k_0t$. Then we record prefixes of lengths $(k_0+1)t$,
$(k_0+2)t$,\ldots\ until we hit a multiple of $t^2$, and so on, until we hit a
multiple of $s$. Then we finally terminate by recording prefixes of lengths multiple 
of $s$. 
\item We work analogously with $t^+$ and $s^+$ for deep nodes whose parents are also 
deep nodes. That is we store all prefixes of lengths $i_\alpha$ , $i_\alpha+1$,~\ldots\
until we hit a length of the form $\ell (\log n)^c+k_0t^+$, then record all prefixes  
of lengths $\ell (\log n)^c+(k_0+1)t^+$,$\ell (\log n)^c+(k_0+2)t^+$ until we hit a length 
of the form $\ell (\log n)^c+k_1{t^+}^2$, 
and so on until we record a length of the form $\ell (\log n)^c+k_{c-2}s^+$. Then we finally store all
prefixes of lengths $\ell (\log n)^c+(k_{c-2}+1)s^+$,$\ell (\log n)^c+(k_{c-2}+2)s^+$\ldots. 
\item For deep nodes with non deep parents, we do the following two things: 
\begin{itemize}
\item We first record short prefixes. We record all strings of lengths $i_\alpha$,$i_\alpha+1$,~\ldots\ 
until we either hit a length multiple of $t$ or length $\ell (\log n)^c$. If we have hit length $\ell (\log n)^c$ 
we stop recording short prefixes. Otherwise, we continue in the same way, with 
prefixes of lengths multiple of $t^u$ for increasing $u=1\ldots{c-1}$ each time terminating the step if we hit a multiple 
of $t^{u+1}$ or completely halt recording short prefixes if we hit length $\ell (\log n)^c$.
\item Secondly, we record long prefixes. That is all prefixes of lengths of the form $\ell (\log n)^c+ks^+$
\end{itemize}
\end{itemize}
Clearly each node
contributes at most $O(ct)$ short prefixes ($O(ct^+)$ long prefixes respectively). In the first case,
there are obviously $O(n\ell^{1/c})$ short prefixes. In the second case, since the
number of deep nodes is at most $2n/(\log n)^2-1$ there are at most
$O(n\ell^{1/c}/\log n)$ long prefixes. Overall, $F^-$ requires
$O(n\ell^{1/c}\log(\ell(\log n)^2))=O(n\ell^{1/c}\log\ell)$ bits, whereas $F^+$
requires $O((n\ell^{1/c}/\log n)\log(\ell n))=O(n\ell^{1/c}\log\ell)$ bits.

The query procedure is modified as follows: for $i=c-1$, $c-2$,~\ldots\,, $0$, if
$p$ is short we consider the string $q=p[0\..|p|-|p|\bmod t^i)$ and check
whether $G(p[0\..F(q)))\geq|p|$. If this happen, we know that $p[0\..F(q))$ is
the name of the exit node of $p$ and we return it. Note that if $p$ is a prefix of some string
in $S$, this must happen at some point (eventually at $i=0$).
If $p$ is long, we do the same using $t^+$ and $s^+$. That is we consider the 
string $q=p[0\..|p|-(|p|-\lceil\ell (\log n)^2\rceil)\bmod {t^+}^i)$ and check
whether $G(p[0\..F(q)))\geq|p|$. If this happen, we know that $p[0\..F(q))$ is
the name of the exit node of $p$ and we return it.

This procedure finds
the name of the exit node of $|p|$ in time $O(|p|/w)$.

\section{Proof of Theorem~\ref{thm:lowerbound}}\label{sec:lowerboundproof}

Let $u=2^\ell$ be the number of possible keys of length $\ell$. We show that there exists a probability distribution on key sets $S$ such that the expected 
space usage is $\Omega((n/k)\log\log(u/n))$ bits. By the ``easy directions of Yao's lemma,'' this implies that the expected space usage 
of any (possibly randomised) data structure on a worst case input is at least $\Omega((n/k)\log\log(u/n))$ bits. 
The bound for $\ell>(1+\varepsilon)\log n$ and $k=1$ follows immediately.

Assume without loss of generality that $n/(k+1)$ and $k$ are powers of $2$. All strings in $S$ will be of the form $abc$, where 
$a\in 2^{\log_2(n/(k+1))}$, $b\in 2^{\ell-\log_2(n/(k+1))-\log_2 k}$, and
$c\in 2^{\log_2 k}$. Let $t=\ell-\log_2(n/(k+1))-\log_2 k$ denote the length of
$b$. For every value of $a$ the set will contain exactly $k+1$ elements: One where $b$ and $c$ are strings of $0$s, and for $b$ chosen uniformly 
at random among strings of Hamming weight $1$ we have $k$ strings for $c\in
2^{\log_2 k}$. Notice that the entropy of the set $S$ is $n/(k+1) \log_2 t$, as we choose $n/(k+1)$ values of $b$ independently from a set of $t$ strings. 
To finish the argument we will need to show that any two such sets require different data structures, which means that the entropy of the 
bit string representing the data structure for $S$ must also be at least $n/(k+1)\log_2 t$, and in particular this is a lower bound on the 
expected length of the bit string.

Consider two different sets $S'$ and $S''$. There exists a value of $a$, and distinct values $b'$, $b''$ of Hamming weight $1$ such that $S'$ 
contains all $k$ $\ell$-bits strings prefixed by $ab'$, and $S''$ contains all $k$ $\ell$-bits strings prefixed by $ab''$. Assume without loss of 
generality that $b'$ is lexicographically before $b''$. Now consider the query for a string of the form $a0^\ell$, which is a prefix of $ab'$ 
but not $ab''$ -- such a string exists since $b'$ and $b''$ have Hamming weight $1$. The number of keys with this prefix is $k+1$ and $1$, 
respectively, for $S'$ and $S''$, so the answers to the queries must be different (both in the multiplicative and additive case). Hence, 
different data structures are needed for $S'$ and $S''$.

\section{Applications in the cache-oblivious model}

\paragraph{Prefix search and counting.}
For prefix search our
constant time results imply an improvement in the cache-oblivious model compared
to results in~\cite{cacheobsd,BFKCOSBT,FGGSCSCCO}. Suppose that the strings are
stored contiguously in increasing lexicographic order. Using Elias--Fano
encoding~\cite{EliESRCASF} we can store pointers to strings such that starting
position of a string can be located in constant time. Using our fastest result, a
weak prefix search takes just $O(|p|/B)$ for the weak prefix search and then just
$O(1)$ time to locate the pointers to beginning of first string in the range
(which we note by $R_s$) and to the ending of the last string in the range (which
we note by $R_e$). Then the prefix search can be carried in optimal $O(K/B)$
I/Os, where $K=R_e-R_s+1$ is the total length of the strings satisfying the query
or in optimal $O(|p|/B)$ I/Os if no string satisfies the query. That is we read
the first $|p|$ bits of first string in the range(or any other string in the
range), just to check whether those $|p|$ bits are identical to $p$, in which
case, we read the remaining $K-|p|$ bits satisfying the search query.

For counting queries we also clearly have the optimal $O(|p|/B)$ I/Os bound as we
can just do a weak prefix search in $O(|p|/B)$ I/Os, locating a range of strings
and then retrieving the first $|p|$ bits of any key in the range.

\paragraph{Range emptiness.}
In the cache oblivious model a range query for an interval $[a\..b]$ of variable
length strings can be done in optimal $O(K/B+|a|/B+|b|/B)$ I/Os where $K$ is size
of output. For that, we can slightly modify the way strings are stored in memory.
As before, we store strings in increasing lexicographic order. but this time we
store the length of each string at the end and at the beginning of each string.
We note that storing the string lengths using Elias $\delta$ or $\gamma$
coding~\cite{EliUCSRI} never increases the lengths of any individual string by
more than a constant factor and that they occupy no more than constant number of 
of memory words each (moreover total space wasted by stored lengths is $O(n\ell)$ bits). Thus we
can scan the strings backward and forward in optimal time in cache-oblivious
model. That is the number of blocks we read is no more than a constant factor
than necessary.

The algorithm becomes very clear now : given the interval $[a\..b]$, as before, we either perform a 
single prefix query for prefix $a$ in case $a$ is a prefix of $b$ or otherwise perform two prefix queries 
for $p0$ and $p1$ (where $p$ is the longest common prefix). 
Then using Elias-Fano we can locate the pointers to first string in first case 
or locate the two pointers to last string in first interval and first string in second interval 
in the second case. Then we only need to scan forward and backward 
checking each time whether the strings are in the interval. Each
checking needs only to read $\max(|a|,|b|)$ bits. Thus we do not perform more
than $O(K/B+|a|/B+|b|/B)$ I/Os. In particular checking for the two boundary
strings does not take more than $O(|a|/B+|b|/B)$ I/Os. Even reading the length of
the two boundary strings takes no more than $O(1)$ I/Os as the lengths are
encoded in no more than a constant number of words and we have $w<B$.
\section{Extensions to larger alphabets}\label{sec:alpha_ext}
Our data structures for for weak prefix search can easily be extended to work for any integer alphabet of size $\sigma$. 
%(without loss of generality, we assume that $\sigma<ln$)
\subsection{First extension approach}
The most straightforward approach consists in considering each character in the alphabet as a sequence of $\log(\sigma)$ bits. The length of a string of length $p$ becomes $|p|\log\sigma$, and the total length of the strings in the collection becomes $O(nl\log\sigma)$. In this context query times for weak prefix search become: 
\begin{itemize}
\item The query time for the space optimal solution becomes $O((|p|\log\sigma)/w+\log(|p|\log \sigma))=O((|p|\log\sigma)/w+\log|p|+\log\log \sigma)$, while the space usage becomes $\HT+O(n(\log\log\ell+\log\log\log\sigma))=\HT+O(n\log\log\ell)$, where $\HT$ is the hollow trie size of the compacted binary trie built on the transformed set of strings.
\item The query time of the time optimal solution becomes $O((|p|\log\sigma)/w$, while space usage becomes $O(n(\ell\log\sigma)^{1/c}\log(\ell\log\sigma))$. Note that query time is still optimal in this case. 
\end{itemize}
\paragraph{Cache oblivious model}
The query times for prefix range and general range queries in the cache oblivious model become : 
\begin{itemize}
\item For the the space optimal solution, query time for a prefix range query becomes $(|p|\log\sigma)/B+\log|p|+\log\log\sigma+K\log(\sigma)/B)$ for a prefix search on a prefix $p$,where $K$ is length of the output in number of characters. For a range query on an interval $[a,b]$ the query time becomes $((|a|+|b|)\log\sigma)/B+\log |p|+\log\log\sigma+K\log(\sigma)/B$. 
\item The query time of the time optimal solution becomes $O(|p|\log\sigma)/B+K\log(\sigma)/B)$ for a prefix query and $((|a|+|b|)\log\sigma)/B+\log |p|+\log\log\sigma+K\log(\sigma)/B$ for range queries.
\end{itemize}

\subsection{Second extension approach}
We now describe the second extension approach which is almost as simple as the first one. We do a small modification of our scheme which will permit us to gain time or space compared to the naive solution above. We first notice that the procedure which identifies exit nodes does not depend at all on the alphabet being binary. Our solution will be built on a top of a compacted trie (which we note by $T_\sigma$)built on the set of strings over original alphabet (we do not convert the strings to binary alphabet). Thus any node in the internal node in the trie will have between $2$ and $\sigma$ children where each children is labeled with a character in $[0,\sigma-1]$. All the data structures defined in section 4 can be reused without any modification on an alphabet of size $\sigma$ instead of an alphabet of size $2$. The range locator in section 5 also extends directly to larger alphabets. A more detailed description of the range locator extension to larger alphabets is given below. Before that we will first present an analysis of space usage of our data structures. For that we first redefine the hollow z-fast trie size for larger alphabets (which we not by $\HT_\sigma(S)$) as : 
\[
\mathrm{\HT_\sigma}(S) = \sum_\alpha( \operatorname{bitlength}(|c_\alpha|) + 1) - 1.
\]
where the summation is done over internal nodes only. Note that for a given data set on alphabet $\sigma$, $\HT_\sigma$ could be much smaller than $\HT$, for two reasons : 
\begin{itemize}
\item The first reason is that the number of internal nodes in $T_\sigma$ could be as small as $O(n/\sigma)$. In contrast the trie $T$ has exactly $n$ internal nodes.
\item The second reason is that the length $|c_\alpha|$ is expressed in number of characters instead of number of bits. That is ${bitlength}(|c_\alpha|)$ uses $\log\log\sigma$ bits.
\end{itemize}
Thus space used by the functions which maps prefixes to their exit nodes can be restated as $\HT_\sigma+O(n\log\log \ell)$ bits for the z-fast hollow trie and $O(n\ell^{1/c}(\log\ell+\log\log n))$ for the time optimal solution. The difference between space usage in this variant and the variant for binary alphabet is the $(\log\log n)$ term. This term comes from the function $F^-$ which dominates space usage and uses $O(\log\ell+\log\log n)$ bits per inserted prefix. This term was previously absorbed by the $\log\ell$ term as for binary alphabet where we had $\ell\geq\log n$.

\paragraph{Range locator for larger alphabets}
The adaptation of range locator for large alphabets is also straightforward. The main difference between range locator for binary alphabet and the range locator for an alphabet of size $\sigma$ is that the names of exit nodes are now of the form $pc$, where $p$ is the extent of some node $\alpha$ in the trie and $c$ is a character from $[0\ldots\sigma)$ instead of $\{0,1\}$. For each node whose name is $pc$ we will store in the mmphf the names $pc$ and $(pc)^+$, where $(pc)^+$ is define recursively in the following way: if $c$ is the last character in the alphabet then $(pc)^+=p^+$, otherwise $(pc)^+=pc^+$, where $c^+$ is the successor of $c$ in the lexicographic order of the alphabet. 
The number of elements inserted in the mmphf is at most $2(2n-1)$. This is the case because we have $2n-1$ for each node contributes at most two prefixes to the range locator. Note that because the two prefixes associated with two different nodes may overlap, we could have less $2(2n-1)$ prefixes.
\\
Depending on the mmphf we used we get two different tradeoffs: 
\begin{itemize}
\item Space usage $O(n(\log\log\ell+\log\log\log\sigma))$ bits with query time $O((|P|\log\sigma)/w+\log(|P|+\log\log(\sigma)))$ if we use logarithmic time mmphf.
\item Space usage $O(n(\log\ell+\log\log\sigma))$ bits with query time $O((|P|\log\sigma)/w)$ if we use constant time mmphf. 
\end{itemize}
\paragraph{Improved space usage for constant time exit node mapping}
The space usage of the constant time exit node mapping functions (functions $F$ and $G$) can be reduced from $O(n\ell^{1/c}(\log\ell+\log\log n))$ to $O(n\ell^{1/c}(\log\ell+\log\log\log n))$. For that we will use three function $F^1$, $F^2$ and $F^3$ instead of just two function $F^-$ and $F^+$. That is $F^1$ stores prefixes of lengths less $\ell\log\log n$, $F^2$ stores prefixes of lengths between $\ell\log\log n$ and $\ell\log n$ and finally $F^3$ stores prefixes of lengths above $\ell\log n$. 
It can easily be seen that space usage is dominated by the function $F^1$ which will thus use $\log\ell+\log\log\log n$ bits of space per stored prefix. We could further reduce space usage by using four or more function, but we will see in next section that this does not give any asymptotic improvement.  

\subsection{Putting things together}
The second extension makes it possible to improve the space/time tradeoff for the logarithmic and constant time weak prefix search solutions.  
\paragraph{Constant time queries}
Using the second extension, we note that space usage of constant time solution can be improved to $O(n(\log\ell+\log\log\sigma))$ in case $\ell$ is constant. Note that because we have $\ell\log\sigma\geq\log n$, we have that $\log\sigma=\Omega(\log n)$. The space usage for the first extension method is at least $\Omega((\ell\log\sigma)^{1/c}(\log\ell+\log\log\sigma))$. Thus with the second extension we have a logarithmic improvement in space usage compared to the first extension. 
However, we can do even better. That is even if we have $\ell=O((\log\log n)^k)$ for any constant $k$ (note that this implies that $\log\sigma=\Omega(\log n/(\log\log n)^k)$), we can still obtain $O(n(\log\ell+\log\log\sigma))$ bits of space. That is by choosing $c=k+1$, the space usage of exit node mapping functions in second extension (the functions $F^1$, $F^2$, $F^3$ and $G$) method becomes $O(\ell^{1/c}(\log \ell+\log\log\log n))=O((\log\log n)^{k/(k+1)}\log\log\log n)$ bits and thus the space usage becomes dominated by the range locator which uses $O(n(\log\ell+\log\log\sigma))$ bits of space. 
\paragraph{Logarithmic time queries}
For logarithmic weak prefix searches, we can have two kinds of improvements. Either space usage improvement or query time improvement. That is by combining the second approach method with the range locator suitable for large alphabets, we get the following:
\begin{itemize}
\item By using a constant time mmphf as a sub-component of the range locator we get query time $O((|P|\log\sigma)/w+\log(|P|))$ with $O(n(\log \ell+\log\log\sigma)$ of space usage.
\item By using the logarithmic time mmphs as a sub-component of the range locator, we get query time $O((|P|\log\sigma)/w+\log(|P|)+\log\log\sigma)$ with improved $\HT_\sigma+O(n\log\log\ell)$ bits of space usage.
\end{itemize}

\subsection{Cache oblivious model}
The performance of prefix search naturally improves in the cache-oblivious model for the logarithmic time prefix search. Using the second extension combined with range locator based constant time mmphf, we get query time $O(|P|/B+\log|P|)$ against $O(|P|/B+\log|P|+\log\log\sigma)$ for the first extension method. With the first extension the used space is $\HT+O(n\log\ell)$ against $O(n(\log\ell+\log\log\sigma))$ for the second extension. Note that the space usage of the first extension is better than that of second extension as we have $n(\log\ell+\log\log\sigma)\geq HT$. Note that the improvement in query time can be very significant in case $\sigma$ is very large. 
\\
For constant time weak prefix queries, the query time is already optimal. However the space/time tradeoff can be improved using the second extension method. For a given constant $c$, the query time is $O(|P|/B+c)$ for both the first and second extension, but space is for second extension is $O(n(\ell^{1/c}(\log\ell+\log\log\log n)+\log\log\sigma))$ against $O(n(\ell\log\sigma)^{1/c}\log(\ell\log\sigma))$ for the first extension method. It can be seen that for large $\log\sigma$ the space saving can be significant. In the particular case where $\log\sigma=\Omega(\log n/(\log\log n)^k)$ for some constant $k$ and the $\ell=(\log\log n)^k$ we can obtain $O(|P|/B+k)$ using only $O(n(\log\ell+\log\log\sigma))$ bits of space.

\end{document}